\documentclass[a4paper,10pt]{article}
\usepackage{amsfonts, amsmath, amssymb, amsthm}
\usepackage[T1]{fontenc} %
\usepackage[latin1]{inputenc} %
\usepackage{xspace}

\usepackage{vmargin}
\setmarginsrb{1.4in}{1.1in}{1.4in}{1in}{0cm}{0mm}{0cm}{8mm}

\theoremstyle{plain}
  \newtheorem{theorem}{Theorem}
  \newtheorem{lemma}[theorem]{Lemma}
  \newtheorem{corollary}[theorem]{Corollary}

\theoremstyle{definition}

\theoremstyle{remark}

\newcommand{\MIST}{\textsc{MIST}\xspace}
\newcommand{\pIST}{\textsc{$p$-IST}\xspace} 
\newcommand{\MISTlong}{\textsc{Maximum Internal Spanning Tree}\xspace}
\newcommand{\internal}[2]{\mathsf{i}_{#1}(#2)}

\title{A Linear Vertex Kernel for\\ Maximum Internal Spanning Tree \thanks{%
A preliminary version of this paper appeared in the proceedings of ISAAC 2009 \cite{FominGST09}.}}

\author{Fedor V. Fomin\thanks{Department of Informatics, University of Bergen, Bergen, Norway.
Email: \texttt{fomin@ii.uib.no}}
\and Serge Gaspers\thanks{Institute of Information Systems (184/3), Vienna University of Technology, Vienna, Austria.
Email: \texttt{gaspers@kr.tuwien.ac.at}}
\and Saket Saurabh\thanks{Institute of Mathematical Sciences, Chennai, India.
Email: \texttt{saket@imsc.res.in}}
\and St\'ephan Thomass\'e\thanks{LIP, ENS Lyon, Lyon, France.
Email: \texttt{stephan.thomasse@ens-lyon.fr}}}

\date{}

\begin{document}

%
%
%
%

\maketitle

\begin{abstract} 
We present a polynomial time algorithm that for any graph $G$ and integer $k\geq 0$, either finds a spanning tree with at least $k$ internal vertices, or
outputs a new graph $G_R$ on at most $3k$ vertices and an integer $k'$ such that $G$ has a spanning tree with at least $k$ internal vertices if and only if $G_R$ has a spanning tree with at
least $k'$ internal vertices. In other words, we show that the \MISTlong problem parameterized by the number of internal vertices $k$ has a 
$3k$-vertex kernel. Our result is based on an innovative application of a classical min-max result 
about hypertrees in hypergraphs which states that ``a hypergraph $H$ contains a 
hypertree  if and only if $H$ is partition connected.''
 
\end{abstract}
%


\section{Introduction}

In the \MISTlong problem  (\MIST), we are given a graph $G $   and the task is to find a spanning tree of $G$ with a maximum number of internal vertices.  \MIST is a natural generalization of 
the {\sc Hamiltonian Path} problem because an $n$-vertex graph has a Hamiltonian path if and only if it has a spanning tree with $n-2$ internal vertices.
%

In this paper we study a parameterized version of \MIST. 
Parameterized decision problems are defined by specifying the input ($I$), the parameter ($k$), and
the question to be answered.  A parameterized problem that can be solved in time $f(k)|I|^{O(1)}$, where $f$ is a
function of $k$ alone is said to be \emph{fixed parameter tractable} (FPT).  The natural parameter $k$ for \MIST is the number of internal vertices in the spanning tree and 
 the parameterized version of \MIST, $p$-\textsc{Internal Spanning Tree}  or \pIST for short, is for a given graph $G$ and integer $k$, decide if $G$ contains a spanning tree with at least $k$ internal vertices. 
 It follows from Robertson and Seymour's Graph Minors theory that \pIST is  FPT \cite{RobertsonS85a}.  Indeed, the property of not having a spanning tree with at least $k$ internal vertices is closed under taking minors, and thus such graphs  can be characterized by a finite set of forbidden minors.  
 One of the consequences of the Graph Minors theory is that every graph property characterized by a finite set of forbidden minors is FPT, and thus  \pIST is FPT. 
 These arguments are however not constructive.  The first constructive algorithm for \pIST is due to Prieto and Sloper  \cite{PrietoS05} and has running time
$2^{4k\log{k}} \cdot n^{O(1)}$. Recently this result was improved by 
Cohen et al.~\cite{CohenF09} who solved a more general directed version of the problem in time $49.4^k \cdot n^{O(1)}$.
When the input graph is cubic, \pIST can be solved in time $2.1364^kn^{O(1)}$ \cite{FernauGR09}.
In this paper we study \pIST from the kernelization viewpoint. 

A parameterized problem is said to admit a {\it polynomial kernel} if there is a polynomial time algorithm (where the degree of the polynomial is independent of $k$), called a {\em kernelization} algorithm, that reduces the input instance to an instance whose size is bounded by a polynomial $p(k)$ in $k$, while preserving the answer. This reduced instance is called a {\em $p(k)$ kernel} for the problem.
Let us remark that the instance size and the number of vertices in the instance may be different,  
and thus for bounding the number of vertices in the reduced graph, the term {\em $p(k)$-vertex kernel} is often used. 
 While many problems on graphs are known to have polynomial kernels  (parameterized by the solution size), there are not so many $O(k)$, or linear-vertex kernels
known in the literature.  Notable examples include a $2k$-vertex kernel for {\sc Vertex Cover}~\cite{CKJ01},
a $k$-vertex kernel for {\sc Set Splitting}~\cite{LokshtanovS09}, and a $6k$-vertex kernel for 
{\sc Cluster Editing}~\cite{Guo09}.

No linear-vertex kernel for \pIST was known prior to our work.  
Prieto and Sloper~\cite{PriSlo2003} provided an $O(k^3)$-kernel for the problem and then improved 
it to $O(k^2)$ in~\cite{PrietoS05}. 
The main result of this paper is that  \pIST has a $3k$-vertex kernel.
The kernelization of Prieto and Sloper is based on the so-called  ``Crown Decomposition Method''~\cite{Abu-KhzamFLS07}.  Here, we use a different method, based on a 
min-max characterization of hypergraphs containing hypertrees by Frank et al.~\cite{FrankKK03}. 
 As a corollary of the new kernelization, we obtain an algorithm for solving \pIST running in time 
 $8^k \cdot n^{O(1)}$.

The paper is organized as follows. In Section~\ref{sec:prelim}, we provide necessary definitions and facts about graphs and  hypergraphs. In Section~\ref{sec:kern_algorithm}, we give the kernelization algorithm. 
Section~\ref{sec:proofoflemma} is devoted to the proof of the main combinatorial lemma, which is central to the 
correctness of the kernelization algorithm. 

\section{Preliminaries}\label{sec:prelim}

\subsection{Graphs}
Let $G=(V,E)$ be an undirected simple graph with vertex set $V$ and edge set $E$. 
For any nonempty subset
$W\subseteq V$, the subgraph of $G$ induced by $W$ is denoted by
$G[W]$.  The \emph{neighborhood} of a vertex $v$ in $G$ is $N_G(v)=\{u\in
V:~\{u,v\}\in E\}$,  and for a vertex set
$S \subseteq V$ we set $N_G(S) = \bigcup_{v \in S} N_G(v)\setminus  S$.
The degree of vertex $v$ in $G$ is $d_G(v)=|N_G(v)|$.
Sometimes, when the graph is clear from the context, we omit the subscripts. 

\subsection{The Hypergraphic Matroid}
\label{sec:hyper} 
Let $H=(V,E)$ be a hypergraph. A hyperedge $e\in E$ is a subset of $V$.  A subset $F$ of hyperedges is a \emph{hyperforest} if $|\cup F'|\geq |F'|+1$ for every subset $F'$ of $F$, where $\cup F'$ denotes 
the union of vertices contained in the hyperedges of $F'$. This condition 
is also called the \emph{strong Hall condition}, where \emph{strong} stands for the extra plus one added
to the usual Hall condition. A hyperforest with $|V|-1$ edges
is called a \emph{hypertree}. Lorea proved (see~\cite{FrankKK03} or~\cite{Lorea75}) that ${\cal M}_H =(E, {\cal F})$,
where $\cal F$ consists of the hyperforests of $H$, is a matroid, called the 
{\it hypergraphic matroid}. Observe that these definitions are well-known when restricted  
to graphs. 

\emph{Shrinking} a hyperedge $e$ means replacing it by another hyperedge $e'\subseteq e$. This operation may create multi-hyperedges.
Lov\'asz~\cite{Lovasz70} proved that $F$ is a hyperforest if and only if every hyperedge $e$ of
$F$ can be shrunk to an edge $e'$ (that is, $e' \subseteq e$ contains two vertices of $e$) in such a way 
that the set $F'$ consisting of these shrunken edges forms a forest in the usual sense, i.e., a simple acyclic graph. 
Observe that if $F$ is a hypertree then its set of shrunken edges $F'$ forms a spanning tree on $V$. 

The {\it border} of a partition ${\cal P}=\{V_1,\dots ,V_p\}$ of $V$   
is the set $\delta ({\cal P})$ of hyperedges of $H$ which intersect at least 
two parts of $\cal P$. A hypergraph is \emph{partition-connected} when 
$|\delta ({\cal P})| \geq |{\cal P}|-1$ for every partition $\cal P$ of $V$.  
The following theorem can be found in~\cite[Corollary 2.6]{FrankKK03}.  

\begin{theorem}\label{thm:hypertree}
$H$ contains a hypertree if and only if $H$ is partition-connected.
\end{theorem}

The proof of Theorem~\ref{thm:hypertree} 
can be turned into a polynomial time algorithm, that is, given a hypergraph $H=(V,E)$ we  
can either find a hypertree or find a partition  $\cal P$ of $V$ such that $|\delta ({\cal P})| < |{\cal P}|-1$ in polynomial time. 
For the sake of completeness, we briefly describe a polynomial time algorithm to do this. We do not attempt to optimize the running time of this algorithm.
 Recall that ${\cal M}_H =(E, {\cal F})$,  
where $\cal F$ consists of the hyperforests of $H$, is a matroid 
 and hence we can construct a hypertree, if one exists, greedily. We start with an empty forest and iteratively 
 try to grow our current hyperforest by adding new edges. When inspecting a new edge we either reject or 
 accept it in our current hyperforest depending on whether by adding it we still have a hyperforest. The only question
is to be able to test efficiently if a given collection of edges forms a hyperforest. 
In other words, we have to check whether the strong Hall condition holds. This can be done 
in polynomial time 
by simply running the well-known polynomial time 
algorithm for testing the usual Hall condition for every subhypergraph $H\setminus v$, where $v$ is a vertex 
and $H\setminus v$ is the hypergraph containing all hyperedges $e\setminus v$ for $e\in E$.  

We can also shrink the edges of a hypertree to 
a spanning tree in polynomial time. For this, consider any edge $e$ of the hypertree with more than 
two vertices  (if none exist, we already have our tree). By the result of Lov\'asz~\cite{Lovasz70} mentioned 
above, one of the vertices $v\in e$ can be deleted from $e$ in such a way that we still have a hypertree.  
Hence we just find this 
vertex by checking the strong Hall condition for every choice of $e\setminus v$ where 
$v\in e$. This implies that we need to apply the algorithm to test the strong Hall condition at most $|V|$ times 
to obtain the desired spanning tree. Consequently, there exists a polynomial time algorithm which can 
find a spanning tree whose edges are obtained by shrinking hyperedges of a partition-connected hypergraph. 

We now turn to the co-NP certificate, that is, we want to exhibit a partition $\cal P$ of $V$ such that 
$|\delta ({\cal P})| < |{\cal P}|-1$ when $H$ is not partition-connected. 
The algorithm simply tries to contract every pair of vertices in $H=(V,E)$ 
and checks if the resulting hypergraph is partition-connected. 
When it is not,
we contract the two vertices, and recurse. We stop when the resulting hypergraph 
$H'$ is not partition-connected, and every contraction results in a partition-connected 
hypergraph. Observe then that if a partition $\cal P$ of $H'$ is such that $|\delta ({\cal P})| < |{\cal P}|-1$
and $\cal P$ has a part which is not a singleton, then contracting two vertices of this part
results in a non partition-connected hypergraph. Hence, the singleton partition is the unique 
partition $\cal P$ of $H'$ such that $|\delta ({\cal P})| < |{\cal P}|-1$. This singleton 
partition corresponds to the partition of $H$ which gives our co-NP certificate.

\section{Kernelization Algorithm}\label{sec:kern_algorithm}
Let $G=(V,E)$ be a connected graph on $n$ vertices and $k \in \mathbb{N}$ be a parameter. In this section we describe an algorithm that 
  takes $G$ and $k$ as an input, and in time polynomial in the size of $G$ either solves \pIST, or produces a reduced  graph $G_R$ on at most $3k$ vertices 
and an integer $k' \leq k$, such that $G$ has a spanning tree with at least $k$ internal vertices if and only if $G_R$ has a spanning tree with at least $k'$ internal vertices. 
In other words, we show that \pIST has a $3k$-vertex kernel.
%



The algorithm is based on the following combinatorial lemma, which is interesting on its own. For two disjoint sets $X,Y \subseteq V$, we denote by $B(X,Y)$ the bipartite graph obtained from $G[X\cup Y]$ by removing all edges with both endpoints in $X$ and all edges with both endpoints in $Y$.

\begin{lemma}\label{lemma:combinatorial}
If  $n\ge 3$, and $I$ is an independent set of $G$ of cardinality at least $2n/3$, then  there are nonempty subsets $S\subseteq V\setminus I$ and $L \subseteq I$ such that
\begin{itemize}
\setlength{\itemsep}{-1pt}
\item[(i)] $N(L)=S$, and
\item[(ii)] $B(S,L)$ has a spanning tree such that all vertices of $S$ and $|S|-1$ vertices of $L$ are internal.
\end{itemize}
 Moreover, given a graph on at least $3$ vertices and an independent set of cardinality at least $2n/3$, such subsets can be found in time polynomial in the size of $G$.
\end{lemma}
The proof of Lemma~\ref{lemma:combinatorial} is postponed to
 Section~\ref{sec:proofoflemma}. Now we give the description of  the kernelization algorithm and use Lemma~\ref{lemma:combinatorial} to prove its correctness.  The algorithm consists of the following reduction rules.

\medskip

\begin{description}
 \item[Rule 1] If $n\leq 3k$, then output $G$ and stop. In this case $G$ is a $3k$-vertex kernel.  Otherwise proceed with Rule~2.
\end{description}

\begin{description}
 \item[Rule 2] Choose an arbitrary vertex $v\in V$ and run a DFS (depth first search) from $v$. If the DFS tree $T$ has at least $k$ internal vertices, then the algorithm has found a solution and stops.
 Otherwise,  because $n> 3k$,  $T$ has at least $2n/3 +2$ leaves, and since all leaves but the root of the DFS tree are pairwise nonadjacent, the algorithm has found an independent set of $G$ of cardinality at least $2n/3$. Proceed with Rule~3.
\end{description}

\begin{description}
 \item[Rule~3 (reduction)] Find nonempty subsets of vertices $S,L \subseteq V$ as in Lem\-ma~\ref{lemma:combinatorial}.
Add a vertex $v_S$ and make it adjacent to every vertex in $N(S)\setminus L$ and add a vertex $v_L$ 
and make it adjacent to $v_S$. Finally, remove all vertices of $S\cup L$. 
 Let $G_R = (V_R, E_R)$ be the new graph and $k'=k-2|S|+2$. 
 Go to Rule~1 with  $G:=G_R$ and $k:=k'$. 
 \end{description}

To prove the soundness of Rule~3, we need the following lemma. Here, $S$ and $L$ are as in Lemma~\ref{lemma:combinatorial}. If $T$ is a tree and $X$ a vertex set, we denote by $\internal{T}{X}$ the number of vertices of $X$ that are internal in $T$.

\begin{lemma}\label{lem:goodopt}
If $G$ has a spanning tree with $k$ internal vertices, then $G$ has a spanning tree with at least $k$ internal vertices in which all the vertices of $S$ and exactly $|S|-1$ vertices of $L$ are internal.
\end{lemma}
\begin{proof}
 Let $T$ be a spanning tree of $G$ with $k$ internal vertices.
 Denote by $F$ the forest obtained from $T$ by removing all edges incident to $L$.
 Then, as long as $2$ vertices $u,v \in S$ are in the same connected component in $F$, remove an edge from $F$ that is incident to one of these two vertices and belongs to the $u$--$v$ path in $F$.
 Observe that in $F$, each vertex from $V\setminus (L\cup S)$ is in the same connected component as some vertex from $S$. Indeed, we only removed an edge $uw$ incident to a vertex
 $w\in V\setminus (L\cup S)$ in case $u,v\in S$ and there was a $u$--$v$ path containing $w$. After removing $uw$, $w$ is still in the same connected component as $v$.
 Now, obtain the spanning tree $T'$ by adding the edges of a spanning tree of $B(S,L)$ to $F$ in which all vertices of $S$ and $|S|-1$ vertices of $L$ are internal (see Lemma~\ref{lemma:combinatorial}).
 Clearly, all vertices of $S$ and $|S|-1$ vertices of $L$ are internal in $T'$. It remains to show that $T'$ has at least as many internal vertices as $T$.

Let $U:=V\setminus(S\cup L)$. Then, we have that $\internal{T}{L} \le \sum_{u\in L} d_T(u) - |L|$ as every vertex in a tree has degree at least $1$ and internal vertices have degree at least $2$. We also have $\internal{T'}{U} \ge \internal{T}{U}-(|L|+|S|-1-\sum_{u\in L}d_T(u))$ as at most $|S|-1-(\sum_{u\in L}d_T(u)-|L|)$ edges incident to $S$ are removed from $F$ to separate $F\setminus L$ into $|S|$ connected components, one for each vertex of $S$.
Thus,
\begin{eqnarray*}
\internal{T'}{V} &=& \internal{T'}{U} +  \internal{T'}{S\cup L} \\
    & \geq & \internal{T}{U}-(|L|+|S|-1-\sum_{u\in L}d_T(u)) + \internal{T'}{S\cup L}\\
    & = & \internal{T}{U} + (\sum_{u\in L}d_T(u) - |L|) -|S|+1 + \internal{T'}{S\cup L} \\
    & \geq &  \internal{T}{U} + \internal{T}{L} -|S|+1 + \internal{T'}{S\cup L} \\
    & = &  \internal{T}{U} + \internal{T}{L}  -(|S|-1) + (|S|+|S|-1)\\
    & = & \internal{T}{U} + \internal{T}{L} +|S| \\
    & \geq & \internal{T}{U} + \internal{T}{L} + \internal{T}{S}  \\
    & = & \internal{T}{V}. 
 \end{eqnarray*}
This finishes the proof of the lemma.
\end{proof}

\begin{lemma}\label{lem:correct}
 Rule~3 is sound,  $|V_R| < |V|$, and $k'\leq k$. 
\end{lemma}
%
\begin{proof}
We claim first that the resulting graph $G_R = (V_R, E_R)$ has a spanning tree with at least $k'=k-2|S|+2$ internal vertices if and only if the original graph $G$ has a spanning tree with at least $k$ internal vertices. 
Indeed, assume $G$ has a spanning tree with $\ell \ge k$ internal vertices. Then, let $B(S,L)$ be as in Lemma~\ref{lemma:combinatorial} and $T$ be a spanning tree of $G$ with $\ell$ internal vertices such that all vertices of $S$ and $|S|-1$ vertices of $L$ are internal (which exists by Lemma~\ref{lem:goodopt}). Because $T[S\cup L]$ is connected, every two distinct vertices $u,v \in N_T (S) \setminus L$ are in different connected components of $T\setminus (L\cup S)$. But this means that the graph $T'$ obtained from $T\setminus (L\cup S)$ by connecting $v_S$ to all neighbors of $S$ in $T\setminus (S\cup L)$ is also a tree in which the degree of every vertex in $N_G(S)\setminus L$ is unchanged. The graph $T''$ obtained from $T'$ by adding $v_L$ and connecting $v_L$ to $v_S$ is also a tree. Then $T''$ has exactly $\ell -2|S| +2$ internal vertices.

In the opposite direction, if $G_R$ has a tree $T''$ with $\ell -2|S| +2$ internal vertices, then all neighbors of $v_S$ in $T''$ are in different components of $T'' \setminus \{v_S\}$. By Lemma~\ref{lemma:combinatorial} we know that  $B(S,L)$ has a spanning tree $T_{SL}$ such that all the vertices of $S$ and $|S|-1$ vertices of $L$ are internal. We obtain a spanning tree $T$ of $G$ by 
considering the forest  $T^*=T'' \setminus \{v_S,v_L\} \cup T_{SL}$ and adding edges between different components to make it connected. For each vertex $u \in N_{T''}(v_S)\setminus \{v_L\}$, add an edge $uv$ to 
$T^*$, where $uv$ is an edge of $G$ and $v\in S$. By construction we know that such an edge always exists. 
Moreover, the degrees of the vertices in $N_G(S)\setminus L$ are the same in $T$ as in $T''$.
Thus $T$ is a spanning tree with $\ell$ internal vertices. 

Finally, as $|S|\geq 1$ and $|L\cup S|\geq 3$, we have that $|V_R| < |V|$ and $k'\leq k$.
\end{proof}

Thus Rule~3 compresses the graph and we conclude with the following theorem.
\begin{theorem}\label{thm:kernel}
\pIST has a $3k$-vertex kernel.
\end{theorem}

\begin{corollary}
\pIST can be solved in time $8^k \cdot n^{O(1)}$.
\end{corollary}
\begin{proof}
Obtain a $3k$-vertex kernel for the input graph $G$ in polynomial time using Theorem~\ref{thm:kernel} and run the $2^n n^{O(1)}$ time algorithm of Nederlof \cite{Nederlof09} on the kernel. 
\end{proof}

\section{Proof of Lemma~\ref{lemma:combinatorial}}\label{sec:proofoflemma}

In this section we provide the postponed proof of Lemma~\ref{lemma:combinatorial}.
Let $G=(V,E)$ be a connected graph on $n$ vertices, $I$ be an independent set of $G$ of cardinality at least $2n/3$ and $C:=V\setminus I$.

Let $Y$ be a subset of $V$. A subset $X\subseteq (V\setminus Y)$ has \emph{$Y$-expansion} $c$, for some $c> 0$, if for each subset $Z$ of $X$, $|N(Z)\cap Y| \ge c \cdot |Z|$.
We first find an independent set $L\subseteq I$ whose neighborhood has $L$-expansion $2$.
For this, we need the following result.

%

\begin{lemma}[\cite{Thomasse09}]\label{thm:expansion}
Let $B$ be a nonempty bipartite graph with vertex bipartition $(X, Y)$ with $|Y| \ge 2|X|$ and such that every vertex of $Y$ has at least one neighbor in $X$. Then there exist nonempty subsets $X' \subseteq X$ and $Y' \subseteq Y$ such that the set of neighbors of $Y'$ in $B$ is exactly $X'$, and such that 
$X'$ has $Y'$-expansion 2. 
Moreover,  such subsets $X',Y'$ can be found in time polynomial in the size of $B$.
\end{lemma}

By using Lemma~\ref{thm:expansion}, we 
find nonempty sets of vertices $S' \subseteq C$ and $L' \subseteq I$ such that $N(L')=S'$ and $S'$ has $L'$-expansion $2$.

\begin{lemma}\label{lem:Sconnected}
Let $G=(V,E)$ be a connected graph and $L'\subseteq V$ be an independent set such that $S'=N(L')$ has $L'$-expansion $2$.
Then there exist nonempty subsets $S\subseteq S'$ and $L\subseteq L'$ such that 
 \begin{itemize}
\setlength{\itemsep}{-1pt}
 \item $B(S,L)$ has a spanning tree in which all the vertices of $L$ have degree at most $2$, 
 \item $S$ has $L$-expansion $2$, and 
 \item $N(L)=S$.
 \end{itemize}
 Moreover, such sets $S$ and $L$ can be found in time polynomial in the size of $G$.
\end{lemma}
\begin{proof}
The proof is by induction on $|S'|$. If $|S'|=1$, the lemma holds with $S:=S'$ and $L:=L'$.
 Let $H=(S',E')$ be the hypergraph with edge set $E'=\{N(v) \mid v\in L'\}$. If $H$ contains a hypertree, then it has $|S'|-1$ hyperedges and we can obtain a tree  $T_{S'}$ on $S'$ by shrinking
edges. We use this to find a subtree $T'$ of $B(S',L')$ spanning $S'$ as follows: for every edge $e=uv$ of $T_{S'}$ there exists a hyperedge corresponding to it and hence a unique vertex, say $w$, in $L'$; we 
delete the edge $e=uv$ from $T_{S'}$ and add the edges $wu$ and $wv$ to $T_{S'}$. Observe that the resulting subtree $T'$ of $B(S',L')$ has the property that every vertex in $T'$ which is in 
$L'$ has degree $2$ in it. Finally, we extend $T'$ to a spanning tree of $B(S',L')$ by adding the remaining vertices of $L'$ as pending vertices. All this can be done in polynomial time using the 
algorithm in Section~\ref{sec:hyper}. 
Thus $S'$ and $L'$ are the sets of vertices we are looking for.
 Otherwise, if $H$ does not contain a hypertree, then $H$ is not partition-connected by Theorem~\ref{thm:hypertree}. Then we can find a partition $\mathcal{P}=\{P_1,P_2,\ldots,P_\ell\}$ of $S'$ such that its border $\delta(\mathcal{P})$ contains at most $\ell-2$ hyperedges of $H$ in polynomial time. Let $b_i$ be the number of hyperedges completely contained in $P_i$, where $1\le i\le \ell$. Then there exists an $i, 1\le i\le \ell,$ such that $b_i \ge 2|P_i|$. Indeed, otherwise $|L'|\le (\ell-2)+\sum_{i=1}^\ell (2|P_i|-1)<2|S'|$, which contradicts
 that $S'$ has $L'$-expansion $2$. Let $X:=P_j$ and  
 $Y:=\{w\in L'|~N(w)\subseteq P_j\}$. 
We know that $|Y|\geq  2|X|$ and hence by Lemma~\ref{thm:expansion} there exists a 
$S^*\subseteq X$ and $L^*\subseteq Y$ such that $S^*$ has $L^*$-expansion $2$ and $N(L^*)=S^*$.
Furthermore, $S^*$ and $L^*$ can be computed from $X$ and $Y$ in polynomial time.
Thus, by the induction assumption, there exist $S \subseteq S^*$ and $L \subseteq L^*$ with the desired properties.
\end{proof}

Let $S$ and $L,$ be as in Lemma~\ref{lem:Sconnected}.
We will prove in the following that there exists a spanning tree of $B(S,L)$ such that all the vertices of $S$ and exactly $|S|-1$ vertices of $L$ are internal. Note that there cannot be more than $2|S|-1$ 
internal vertices in a spanning tree of $B(S,L)$ without creating cycles. By Lemma~\ref{lem:Sconnected}, we know that there exists a spanning tree of $B(S,L)$ in which $|S|-1$ vertices of $L$ have degree exactly 
$2$.

Consider the bipartite graph $B_2$ obtained from $B(S,L)$ by adding a copy $S_c$ of $S$ (each vertex in $S$ has the same neighborhood as its copy in $S_c$ and no vertex of $S_c$ is adjacent to a vertex in $S$). As $|L|\ge |S \cup S_c|$ and each subset $Z$ of $S \cup S_c$ has at least $|Z|$ neighbors in $L$, by Hall's theorem, there exists a matching in $B_2$ saturating $S\cup S_c$. This means that in $B(S,L)$, there exist two edge-disjoint matchings $M_1$ and $M_2$, both saturating $S$. We refer to the edges  from $M_1\cup M_2$ as the \emph{favorite} edges.

\begin{lemma}\label{lem:SLinternal}
 $B(S,L)$ has a spanning tree $T$ such that all the vertices of $S$ and $|S|-1$ vertices of $L$ are internal in $T$.
\end{lemma}
\begin{proof}
 Let $T$ be a spanning tree of $B(S,L)$ in which all vertices of $L$ have degree at most $2$, obtained using Lemma~\ref{lem:Sconnected}. As $T$ is a tree, exactly $|S|-1$ vertices of $L$ have degree $2$ in $T$. As long as a vertex $v\in S$ is not internal in $T$, add a favorite edge $u v$ to $T$ which was not yet in $T$ ($u\in L$), and remove an appropriate edge from the tree which is incident to $u$ so that $T$ remains a spanning tree. Vertex $v$ becomes internal and the degree of $u$ in $T$ remains unchanged. As $u$ is only incident to one favorite edge, this rule increases the number of favorite edges in $T$ even though it is possible that some other vertex in $S$ would have become a leaf.  We apply this rule until no longer possible. We know that this rule can only be applied at most $|S|$ times. In the end, all the vertices of $S$ are internal and $|S|-1$ vertices among $L$ are internal as their degrees remain the same.
\end{proof}

To conclude with  the proof of Lemma~\ref{lemma:combinatorial}, we observe that $S \subseteq C$, $L \subseteq I$ and $N(L)=S$ by the construction of $S$ and $L$, and by Lemma~\ref{lem:SLinternal}, $B(S,L)$ has a spanning tree in which all the vertices of $S$ and $|S|-1$ vertices of $L$ are internal.

\section{Conclusion}\label{sec:conclusion}

We have designed a polynomial time algorithm transforming an instance $(G,k)$ for \MIST
into an equivalent instance $(G',k')$ for \MIST such that $G'$ has at most $3k$ vertices.
Moreover, a spanning tree of $G$ with $k$ internal vertices
can be obtained from a spanning tree of $G'$ on $k'$ internal vertices in polynomial time, by replacing vertices $v_S$ and $v_L$ 
as described by Lemma \ref{lem:correct}.
Since, for an instance $(G,k)$, the reduction rule can be executed whenever a DFS tree of $G$ has at most $|V(G)|/3$ internal nodes,
our algorithm also implies a 3-approximation in polynomial time. The best known approximation ratio
that can be achieved in polynomial time for this problem is currently $5/3$ \cite{KnauerS09}.

We finish by asking whether the kernel size can be improved.

\bibliographystyle{plain}

\end{document}